\documentclass[12pt, draftclsnofoot,onecolumn]{IEEEtran}
\IEEEoverridecommandlockouts
\usepackage{textcomp}
\usepackage{cite}
\usepackage{graphicx}
\usepackage{balance}
\usepackage{color}
\usepackage{url}
\usepackage{graphicx, pstool}
\usepackage{float}
\usepackage{subcaption}
\usepackage{array}
\usepackage[english]{babel} 
\usepackage{amsmath,amssymb,amsfonts}
\usepackage{amsthm}
\usepackage{textcomp}
\usepackage{algorithmic}
\usepackage{caption}
\usepackage{breqn}
\usepackage[ruled,vlined,linesnumbered]{algorithm2e}
\usepackage{cuted}

\normalsize
\makeatletter
\renewcommand{\fnum@figure}{Fig. \thefigure}
\makeatother

\newcolumntype{P}[1]{>{\centering\arraybackslash}p{#1}}
\newtheorem{theorem}{Theorem}

\newtheorem{lemma}{Lemma}

\newtheorem{corollary}{Corollary}

\newtheorem{proposition}{Proposition}

\newtheorem{conjecture}{Conjecture}

\newtheorem{sketch}{Sketch of Proof}

\newcolumntype{C}[1]{>{\centering\arraybackslash}m{#1}}

	\IEEEaftertitletext{\vspace{-1\baselineskip}}

\begin{document}
	
	\title{Frequency Permutation Subsets for Joint Radar and Communication
	}
	
	\author{
		Shalanika Dayarathna,~\IEEEmembership{Member,~IEEE,} Rajitha Senanayake,~\IEEEmembership{Member,~IEEE,}\\ Peter Smith,~\IEEEmembership{Fellow,~IEEE} and Jamie Evans,~\IEEEmembership{Senior Member,~IEEE}
	}
	\maketitle
	
	\begin{abstract}
		This paper focuses on waveform design for joint radar and communication systems and presents a new subset selection process to improve the communication error rate performance and global accuracy of radar sensing of the random stepped frequency permutation waveform. An optimal communication receiver based on integer programming is proposed to handle any subset of permutations followed by a more efficient sub-optimal receiver based on the Hungarian algorithm. Considering optimum maximum likelihood detection, the block error rate is analyzed under both additive white Gaussian noise and correlated Rician fading. We propose two methods to select a permutation subset with an improved block error rate and an efficient encoding scheme to map the information symbols to selected permutations under these subsets. From the radar perspective, the ambiguity function is analyzed with regards to the local and the global accuracy of target detection. Furthermore, a subset selection method to reduce the maximum sidelobe height is proposed by extending the properties of Costas arrays. Finally, the process of remapping the frequency tones to the symbol set used to generate permutations is introduced as a method to improve both the communication and radar performances of the selected permutation subset.
	\end{abstract}
	
	
	\section{Introduction}\label{Sec-Intro}
	In recent years, joint radar and communication systems have emerged as a promising paradigm to simultaneously perform radar sensing and communication using the same time/frequency resource. Applications include a variety of emerging applications such as next generation automobile, intelligent transport systems, drone surveillance, wireless positioning and activity recognition as well as defence applications with integrated battlefields \cite{JCC.2020.01.001,2973976,3070399,Ext_paper14}. 
	
	\subsection{Related Work} 
	In the area of joint radar and communications, a considerable amount of research has focused on the resource allocation and optimization of radar systems and communication systems when the two systems operate separately but share the same resources \cite{7492941,2018.08.016}. While this leads to a better utilization of resources, the separation between the two systems results in lower data rates in communications and reduced detection and estimation performance in radar sensing. On the other hand, the design of joint waveforms that can simultaneously perform wireless communication and remote sensing has gained more interest over the years due to its ability to improve the efficiency of spectrum and energy usage and minimize system size and cost \cite{3070399}. Under the joint waveform design, one category of research focuses on traditional communication waveforms to integrate radar sensing functionality whereas another category of research focuses on traditional radar waveforms where the communication bits are modulated into radar pulses. The waveform design presented in this paper focuses on the second category of research.
	
	Under the first category, protocols such as the wireless local area network standard and the standard for dedicated short-range communication in vehicles have been investigated in \cite{2774762,2758581}. However, these protocols can only support short range sensing in the order of tens of meters. On the other hand, the classical orthogonal frequency division multiplexing (OFDM) waveform has been considered for joint radar sensing and communication \cite{9049009,4977024,6875811,8359797}. In \cite{4977024}, the OFDM waveform is shown to achieve low sidelobes and high Doppler tolerance while maintaining the information transmission capacity. However, the high peak-to-average-power ratio (PAPR) of the OFDM waveform reduces the detection range of radar sensing. Taking a step further, in \cite{6875811,8359797}, the use of multiple-input multiple-output (MIMO) OFDM and/or weighted OFDM is shown to achieve low PAPR thus improving its suitability for radar sensing. Recently, the new orthogonal time frequency space (OTFS) waveform has been proposed for joint radar and communication \cite{8757044,9266546,2998583} and in \cite{8757044}, it is shown that the OTFS waveform provides a similar radar performance to that of the OFDM waveform while achieving a higher communication data rate. In \cite{2956689}, a virtual waveform in the mm-Wave band is proposed where the preamble of the communication frame is exploited as a virtual radar pulse, which is conceptually similar to the staggered pulse repetition intervals (PRI) used in long range radar waveforms. Considering a joint radar and communication network with multiple users, in \cite{2994215}, different bit intervals of multiple users are multiplexed in time, code or frequency domains to generate different radar waveforms. 
	
	Under the second category, information bits are modulated on to frequency hopping waveforms used in traditional radar applications using existing modulation schemes such as phase shift keying (PSK) \cite{7944485}. This was later extended in \cite{8378526}, where the Costas frequency codes are used to implement frequency diverse array (FDA) MIMO combined with PSK in each code. On the other hand, linear frequency modulated (LFM) radar signals have been integrated with spread spectrum \cite{2011.187}, minimum shift keying (MSK) \cite{108B012}, continuous phase modulation (CPM) \cite{7944238} and multi-carrier modulation such as OFDM \cite{8628459} to support communication. Due to limited randomness associated with the waveform parameters, the modulation of communication bits into radar waveforms is challenging and this needs to be done in such a way that the degradation in radar performance is acceptable. Taking a different approach, in \cite{3172111}, a new waveform design based on the stepped frequency radar waveform and permutation coding is proposed. By mapping the natural numbers of the information symbols to different frequency sequences, the information is modulated based on the choice of the waveform, thus maintaining the randomness of the radar waveform without limiting the communication data rate. Further, efficient methods to encode and decode the information bits are presented when all the permutation codes are used for communication. This work is later extended in \cite{2107.14396}, where the PSK based random phase modulation is combined with the permutation based random stepped frequency waveform to further improve the communication data rate. 
	
   \subsection{Motivation}       
   The new waveform design using the random stepped frequency radar waveform for joint radar and communications was first introduced in \cite{3172111} in which the randomness in frequency sequence was used to modulate information. The idea was to consider all possible permutations of frequencies, which we define as the \emph{universal permutation set}, and to generate a set of waveforms corresponding to the universal permutation set. Then, the information is modulated based on the choice of the waveform in each transmission. It was demonstrated that this approach can improve the communication data rate while maintaining good radar performance.   
   The use of the universal permutation set for data communication is desirable based on several aspects. Firstly, it provides $M!$ waveforms to chose from during the data modulation phase, thus achieving improved data rates and good radar performance compared to the LFM radar waveform. Secondly, the receiver implementation can be formulated as an assignment problem, thus enabling the design of an efficient and optimal communication receiver. Thirdly, an efficient encoding scheme can be designed, thus avoiding the use of a large look-up table. While the work in \cite{3172111} and \cite{2107.14396} was restricted by the use of universal permutation set, we generalize this to consider any subset of permutations. We note that the added flexibility can improve either the communication performance or the radar performance as required. Motivated by low rate applications such as the navigation function in vehicle-to-everything (V2X) communications where high communication reliability and good radar performance is required, in this work, we proceed to ask three interesting and important questions.
   \begin{itemize}
   	\item \emph{Question 1: Can the communication and radar performance of this waveform design be improved by selecting a subset of permutations?} We note that under the universal set, the minimum Hamming distance between any two waveforms is two. However, according to permutation coding, the communication error rate is inversely proportional to the minimum Hamming distance \cite{blake1979coding}. As such, by selecting a subset of permutations with a larger minimum Hamming distance, we can improve the communication error rate. We also note that the range and velocity estimation under radar performance depends on the transmitted permutation and that some permutations have better radar performance compared to others \cite{12967}. As such, by only selecting the permutations with good radar performance it is possible to improve the accuracy of range and velocity estimation. While there can be many possible subsets, in this paper, we look into subsets that can be obtained following a systematic approach. 
   	\item \emph{Question 2: Will the efficient Lehmer code based encoding scheme still work when only a subset of permutations are selected for transmission?} When the universal set of permutations was selected, the incoming data symbols are mapped to their corresponding waveforms using a combinatorial transform known as the Lehmer code. Under this approach, the incoming data symbol is first transformed to its natural number according to the factorial number system. Then each natural number is mapped to the rank of the corresponding unique permutation in the lexicographic order. Since, each permutation is mapped to a unique waveform, the transmitted waveform is obtained based on this selected permutation. However, when we only select some of the waveforms for information modulation, the direct mapping between the natural number and the lexicographic rank of the permutation fails. Therefore, in general, a look-up table is required for encoding and decoding purposes. Therefore, in this work, we look into the possible design of efficient encoding schemes under different subsets. 
   	\item \emph{Question 3: Can an efficient communication receiver implementation based on the Hungarian algorithm still be used when only a subset of permutations are selected for transmission?} When the universal set is considered, the optimal receiver implementation simplifies to an assignment problem. As such, the Hungarian algorithm always results in the optimal solution. However, when we only select a subset of permutations for information modulation, the Hungarian algorithm fails. This is because the Hungarian algorithm searches through all possible permutations and outputs the permutation that provides the optimal solution, but the resulting permutation might not be in the selected subset. Therefore, in this work, we investigate a receiver implementation that can be used for any subset of permutations. 
   \end{itemize}    
	   
   \subsection{Contribution and Organization}
   In response to the above questions, the contributions of this paper are as follows.
	\begin{itemize}
		\item We investigate the problem of selecting a subset of permutations in order to improve the communication and radar performances and propose an integer programming (IP) optimization based optimal receiver implementation that can be used under any subset of permutations. We also design a novel sub-optimal receiver based on the Hungarian algorithm that has a worst case computational complexity of $O(M^3)$, when the number of frequency tones used in the waveform is given by $M$. 
		\item Based on the minimum Hamming distance in permutation coding, we propose two methods of selecting a subset to improve the communication error rate. Under each method we propose a new efficient encoding scheme that can be used to map incoming data symbols to their corresponding waveforms without the help of a look-up table. Further, under the block based approach, we further show that it is possible to design an efficient optimal receiver implementation. We also derive the block error rate under each method for both additive white Gaussian noise (AWGN) channels and correlated Rician fading channels. 
		\item Generalizing the property of Costas arrays, we propose another subset selection method to improve the global accuracy of radar sensing by reducing the maximum sidelobe height. We further discuss the feasibility of combining two subsets based on the communication perspective and the radar perspective in order to improve the overall performance of the joint waveform. 
	\end{itemize}
	The validity of the analysis is illustrated using numerical examples. The rest of the paper is structured along the following lines. In Section \ref{Sec-model}, we provide the system model for the joint radar and communication based on the random stepped frequency permutation waveform. Then, we discuss the implementation of the communication receiver followed by a performance analysis in Section \ref{Sec-Rec}. Next, we discuss the selection of permutation subsets to improve communication error rate in Section \ref{Sec-Comm} and radar performance, specifically the global accuracy of target detection, in Section \ref{Sec-Radar}. Finally, numerical examples are presented in Section \ref{Sec-Simu} followed by the conclusions in Section \ref{Sec-Conc}.
	
	\section{System Model}\label{Sec-model}
	In this paper, we consider a joint radar and communication system model consisting of one transmitter, one communication receiver and a geographically separated single moving target. The transmitter sends a joint waveform that is used to detect the moving target while simultaneously performing data communication with the communication receiver. As is commonly assumed in the literature, the transmitter is considered to be a mono-static radar. Hence, the radar transmitter and the receiver are collocated and the radar receiver knows the exact waveform used for target detection at any given time. The range and the velocity of the moving target can be estimated using the return signal. At the communication receiver side, we assume that the received signal from the transmitter can be processed to decode the transmitted information. 

	This work is based on the random stepped frequency waveform proposed for the purpose of joint radar and communication in \cite{3172111}. In the following, we provide only the necessary details of this waveform design for completeness. A random stepped frequency radar waveform is generated using $M$ pulses, each with a pulse width of $T$ seconds. Let us consider $M$ frequency tones $f_0, f_1, ... f_{M-1}$ with spacing $\Delta_f$. The frequency sequence for the above random stepped frequency radar waveform is constructed by randomly arranging these $M$ frequency tones such that each frequency is used only once. For example, in the traditional linear stepped frequency radar waveform the frequencies are picked either in ascending or descending order. When we consider all possible permutations related to the $M$ frequency tones we can construct a set of $M!$ potential waveforms. These are all constant amplitude waveforms and the complex baseband signal of the $i$-th waveform can be expressed as,
	\begin{align}\label{eq1}
	s_i(t) = \sqrt{\dfrac{E}{MT}} \sum_{m=0}^{M-1} s_p(t-mT) \textrm{exp}\bigg(j2\pi f_m^i(t-mT)\bigg),
	\end{align}
	where
	\begin{align*}
	s_p(t) = \left\{
	\begin{array}{ll}
	1 \qquad \qquad & 0 \le t \le T \\
	0 & \textrm{otherwise,}
	\end{array}
	\right.
	\end{align*}
	with $f_m^i$ denoting the frequency relating to the $m$-th index of the $i$-th permutation given as $[f_0^i,f_1^i,...,f_{M-1}^i]$ and $E$ denoting the signal energy which satisfies
	\begin{align*}
	\int_{0}^{MT}s_i^2(t) \, dt = E, \qquad i = 0,1,...,M!-1.
	\end{align*}
	It is also assumed that all $M$ frequencies are orthogonal such that $\Delta_f = q/T$ where $q$ is an integer. 
	
	The information is modulated based on the choice of the waveform using a combinatorial transform known as the Lehmer code. Under this modulation scheme, the incoming data symbol is first transformed to its natural number, using the factorial number system, and then this natural number is mapped to the rank of the corresponding unique permutation in the lexicographic order. For example, if the natural number of the incoming information symbol is $i$, then the waveform $s_i(t)$ is transmitted. Therefore, the number of information symbols that can be modulated, using this system model, is proportional to the number of waveforms available for transmission. 
		
	\section{Performance Analysis in Data Communication}\label{Sec-Rec}
	Let us first assume that an arbitrary subset of waveforms, $\mathbf{\mathcal{S}}$, is selected from the complete set of $M!$ waveforms and focus on the communication performance and the receiver implementation of the joint system. Hence we start our investigation focusing on Question 3 and move to Questions 1 and 2 later in the paper. 
	
	\subsection{Maximum Likelihood Detection}
	We consider the optimal maximum likelihood (ML) detection at the communication receiver consisting of $N$ receive antennas. When the waveform, $s_i(t)\in \mathbf{\mathcal{S}}$, is transmitted, the received signal vector at the communication receiver can be expressed as,
	\begin{align}\label{eq2}
	\mathbf{r}(t) = \mathbf{h} \, s_i(t) + \mathbf{n}(t),
	\end{align}
	where $\mathbf{h}$ is the fast fading channel vector and $\mathbf{n}(t)$ is the AWGN vector. The distribution of $\mathbf{n}(t)$ is a complex Gaussian distribution where each element has a zero mean and variance $N_0$. With the assumption that the channel vector $\mathbf{h}$ is known at the receiver, the detected symbol under a coherent ML detector can be written as \cite{3172111},
	\begin{align}\label{eq4}
	\hat{s_i}(t) = {\underset{s_j(t)\in \mathbf{\mathcal{S}}}{\textrm{arg max}}}\;\mathrm{Re}\bigg(\int_{0}^{MT} s_j^{*}(t) \mathbf{h}^H \mathbf{r}(t) \, dt\bigg),
	\end{align}
	where Re(.) represents the real part of the argument. 
	
	\subsubsection{Optimal receiver implementation}
	From \eqref{eq4}, we can clearly see that under the ML decision rule, the obvious implementation of the optimal receiver involves the computation of the correlation between every potential transmit waveform and the corresponding received waveform to obtain the waveform that produces the highest correlation. We define this method of implementation as the basic optimal receiver. However, such receiver implementation has a complexity linear in the size of the selected permutation subset, which could be quite significant for large $M$. 	
	Therefore, we reformulate the optimal receiver in \eqref{eq4} as an IP optimization problem. We define the correlation matrix for the received signal as,
	\begin{align}\label{eq18}
	\mathbf{R} = (r_{uv}) \in \mathbb{R}^{M\times M},
	\end{align}
	where the $uv$-th element, $r_{uv}$, which is given as \cite{3172111},
	\begin{align}\label{eq19}
	r_{uv} = \mathrm{Re} \biggl(\int_{(u-1)T}^{uT} \mathbf{h}^H\mathbf{r}(t)\psi_v \big(t-(u-1)T\big)\, dt\biggr),
	\end{align}
	denotes the correlation between the received signal, $\mathbf{h}^H\mathbf{r}(t)$, and the basis function  $\psi_v \big(t\!-\!(u\!-\!1)T\big)$ which can be expressed as $\psi_v(t) = \sqrt{2E/T} s_p(t)\, \cos\big(2\pi f_v(t)\big)$. Therefore, $r_{uv}$ denotes the correlation between the $u$-th pulse in the received signal and the $u$-th pulse of any waveform which has the $v$-th frequency tone, $f_v$, transmitted during the $u$-th pulse. As such the optimal receiver under the ML decision rule in \eqref{eq4} is equivalent to selecting a waveform within $\mathbf{\mathcal{S}}$ such that the summation of the correlations over all $M$ pulses is maximized.
	
	When the full set of $M!$ permutations are used, the optimal receiver can be formulated as an assignment problem, where $M$ elements need to be selected to maximize the summation of those $M$ elements such that only one element is selected from each row and each column of $\mathbf{R}$. The Hungarian algorithm can be used to solve this assignment problem. However, when we consider only a subset of permutations the optimal receiver can no longer be solved as an assignment problem because the solution to the general assignment problem may not belong to $\mathbf{\mathcal{S}}$. Therefore, in this paper we reformulate the optimization problem of the optimal receiver as follows.
	\begin{align}\label{eq20}
	&{\underset{X_{uv}}{\textrm{max}}}\;\sum_{u=1}^{M}\sum_{v=1}^{M}  r_{uv} X_{uv} \nonumber \\
	&{\rm{s.t \;\;\; }}
	\sum_{u=1}^{M} X_{uv} = \sum_{v=1}^{M} X_{uv} = 1, \nonumber \\
	& \qquad \sum_{u=1}^{M}\sum_{v=1}^{M} X_{uv} Y_{uv}^{(k)} \ge 1, ~~ \forall k \notin \mathbf{\mathcal{S}},
	\end{align}
	where $Y_{uv}^{(k)}=0$ if the $v$-th element of the $k$-th permutation is $u$ and $Y_{uv}^{(k)}=1$ otherwise. We note that this integer optimization can be solved using any existing IP solver. We also note that while the average complexity of this IP based optimal receiver can be better than the basic optimal receiver, the worst case complexity of any IP solver remains exponential. 
	
	\subsubsection{Efficient sub-optimal receiver implementation}
	Given the computational complexity of the optimal receiver implementation, we next propose a sub-optimal receiver based on the Hungarian algorithm, which is more efficient than the optimal receiver. To do that, we first remove the last constraint in \eqref{eq20} which restricts the permitted permutations. Then the optimal receiver simplifies to an assignment problem, which we can solve using the traditional Hungarian algorithm \cite{Hungarian}. Next, we check whether the detected permutation is within the selected subset $\mathbf{\mathcal{S}}$ and if so the detection process is completed. If not, we generate the neighboring permutations which are included in $\mathbf{\mathcal{S}}$ and maintain a Hamming distance less than or equal to a selected distance $d$ from the detected permutation. Then we compute the correlation of the received signal with each waveform related to these neighboring permutations. For example, under the universal set, each permutation has $M(M-1)/2$ neighboring permutations with Hamming distance $2$ \cite{3172111}. Therefore, if we consider $d=2$ then for each detected permutation which does not belong to $\mathbf{\mathcal{S}}$, the correlation of the received signal with a maximum of another $M(M-1)/2$ waveforms needs to be computed. Finally, the waveform with the highest correlation is selected as the detected waveform. In Algorithm~\ref{Algorithm1}, we outline the main steps of our proposed sub-optimal receiver implementation. We also note that the worst case complexity of a simple binary search to see if a given permutation belongs to $\mathbf{\mathcal{S}}$ or not is $O(\log |\mathbf{\mathcal{S}}|)$, where $|\mathbf{\mathcal{S}}|$ is the cardinality of the selected subset $\mathbf{\mathcal{S}}$. As a result, the overall complexity of our proposed sub-optimal algorithm remains $O(M^3)$ when $d=2$ which is similar to the Hungarian algorithm. However, it is also possible to consider neighboring permutations with larger Hamming distances with a view to improve the performance of Algorithm \ref{Algorithm1}, but at the expense of higher computational complexity.
	\begin{algorithm}
		\DontPrintSemicolon 
		\SetAlgoLined		
		Obtain (-$\mathbf{R}$) by negating the correlation matrix $\mathbf{R}$\;
		$\hat{s}_i(t) \leftarrow$ waveform corresponding to the output of the Hungarian algorithm for (-$\mathbf{R}$)\; 
		\If{$\hat{s}_i(t) \notin \mathbf{\mathcal{S}}$}{
			construct $\mathbf{\mathcal{{S}}}_i$, the set of waveforms corresponding to distance $d$ neighbors of $\hat{s}_i(t)$\\
			\For{$s_j(t) \in \mathbf{\mathcal{{S}}}_i$}{
				Compute the correlation of $s_j(t)$ with $\mathbf{r}(t)$\;
			}
			$\hat{s}_i(t)$ $\leftarrow$ $s_j^{*}(t)$, which is the waveform that results in the highest correlation with $\mathbf{r}(t)$\;
		}	
		\caption{Proposed Sub-Optimal Receiver}
		\label{Algorithm1}
	\end{algorithm}	

	\subsection{Block Error Rate Performance}\label{BLER}
	Assuming that the transmitted data is equally likely, the block error rate (BLER) of decoding the received waveform can be expressed as,
	\begin{align}\label{eq5}
	P_e = \dfrac{1}{|\mathbf{\mathcal{S}}|} \sum_{i=0}^{|\mathbf{\mathcal{S}}|-1}[1-P_c(i)],
	\end{align}
	where $P_c(i)$ denotes the probability of a correct decision when waveform $s_i(t)$ is transmitted. Note that  complex multi-dimensional integrals over the multivariate Gaussian density are required for the exact computation of $P_e$. Therefore, taking a more tractable approach we consider the union bound determined by,
	\begin{align}\label{eq6}
	P_e \le P_e^{UB} = \dfrac{1}{|\mathbf{\mathcal{S}}|} \sum_{i=0}^{|\mathbf{\mathcal{S}}|-1} \sum_{j=0,j \neq i}^{|\mathbf{\mathcal{S}}|-1} P_{ij},
	\end{align}
	where $P_{ij}$ denotes the pairwise error probability of detecting $s_j(t)$ when $s_i(t)$ is transmitted. When we consider the ML detection rule given in \eqref{eq4}, $P_{ij}$ can be expressed as,
	\begin{align}\label{eq7}
	P_{ij} = \textrm{Pr} \biggl[\mathrm{Re}\bigg(\int_{0}^{MT} s_i^{*}(t) \mathbf{h}^H \mathbf{r}(t) \, dt\bigg) < \mathrm{Re}\bigg(\int_{0}^{MT} s_j^{*}(t) \mathbf{h}^H \mathbf{r}(t) \, dt\bigg)\biggr].
	\end{align}
	Using some mathematical manipulations and conditioning on channel fading, \eqref{eq7} can be re-expressed as
	\begin{align}\label{eq9}
		P_{ij} = \textrm{Pr} \biggl[\sqrt{\mathbf{h}^H \mathbf{h}} < \sqrt{\dfrac{N_0 M}{E d(i,j)}} Z\biggr],
	\end{align}
	where $d(i,j)$ denotes the Hamming distance between the permutations corresponding to $s_i(t)$ and $s_j(t)$ and $Z \sim \mathcal{N}(0,1)$. As such, we can see that the probability of detecting $s_j(t)$ when $s_i(t)$ is transmitted depends on the Hamming distance between the permutations corresponding to $s_j(t)$ and $s_i(t)$. In the following, we present analytical expressions for $P_e^{UB}$ under two traditional channel models, namely, the AWGN channel and the correlated Rician fading channel models. 
	
	\subsubsection{AWGN channel}
	Without loss of generality, a unit channel gain is assumed such that $E[\mathbf{h}^H\mathbf{h}]=N$. Whilst not included due to page limitations, following a similar approach to \cite{3172111}, the union bound given in \eqref{eq6} can be simplified to,
	\begin{align}\label{eq8.1}
	P_e^{UB} &= \sum_{l=d_{min}}^{M}  A_l \mathcal{Q}\biggl(\sqrt{\dfrac{NEl}{N_0 M}}\biggr),
	\end{align}	
	where $d_{min}$ is the minimum Hamming distance between any two permutations in the selected subset, $\mathbf{\mathcal{S}}$, $A_l$ represents the number of permutations within $\mathbf{\mathcal{S}}$ with Hamming distance $l$ from a given permutation and $\mathcal{Q}(.)$ represents the Gaussian Q-function.
	
	\subsubsection{Rician fading channel}
	To obtain further understanding about the effect of fading, next we consider the correlated Rician fading model where the relative strength of the line-of-sight (LoS) path to the scattered path is determined by the Rician factor, $K$. Thus, we can express the $N \times 1$ fast fading channel vector as $\mathbf{h} = \sqrt{\frac{K}{K+1}}\Delta +  \sqrt{\frac{1}{K+1}}\mathbf{C}_u^{1/2}\mathbf{u},$ where $\Delta$ denotes the complex LoS phase vector such that the $n$-th element, $|\Delta_n|^2 = 1$, $\mathbf{u} \sim \mathcal{CN}(0,\mathbf{I})$ and $\mathbf{C}_u$ is the $N\times N$ correlation matrix of the scattered component. Whilst not shown here due to page limitations, after some mathematical manipulations we can obtain the union bound under the correlated Rician fading model as \cite{2107.14396},
	\begin{align}\label{eq8.2}
	P_e^{UB} &= \sum_{l=d_{min}}^{M} \dfrac{A_l}{\pi} \int_{0}^{\pi/2} \biggl[\prod_{n=1}^{N}\bigg(\dfrac{\alpha_l\sin^2\theta}{\lambda_{n}+\alpha_l\sin^2\theta}\bigg)\exp\bigg(\sum_{n=1}^{N}\dfrac{-K|(\mathbf{V}^H\Delta)_n|^2}{\lambda_{n}+\alpha_l\sin^2\theta}\bigg)\biggr] d\theta,
	\end{align}
	where $\alpha_l =2N_0M(K+1)/El$, $\mathbf{V}$ is a unitary matrix, $\mathbf{\Omega}=\textrm{diag}(\lambda_1,...,\lambda_N)$ is a diagonal matrix containing the eigenvalues $\lambda_1,...,\lambda_N$ of $\mathbf{C}_u$ such that $\mathbf{C}_u = \mathbf{V}\mathbf{\Omega}\mathbf{V}^H$ and $(\mathbf{V}^H\Delta)_n$ denotes the $n$-th element of the vector $\mathbf{V}^H\Delta$. 

	In this work, we focus on two communication performance measures, namely, the BLER and the data rate. According to permutation coding, the BLER is inversely proportional to the minimum Hamming distance $d_{min}$ \cite{blake1979coding}. The larger the Hamming distance is, the smaller the BLER. Note that under the universal set of all $M!$ permutations, $d_{min}=2$. On the other hand, the communication data rate, which is given by $\log_2(M!)$ under the universal set, is proportional to the number of selected waveforms \cite{3172111}. The larger the number of selected waveforms the larger the communication data rate. As such, there is a clear trade-off between the communication data rate and the BLER.
	
	In applications such as administrative functions, low probability of intercept data transmission and navigation function in V2X communication, it is desirable to select a subset of waveforms which results in low data rate but high communication reliability and good radar performance \cite{JCC.2020.01.001}. Therefore, we next focus on Question 1 and investigate two different methods to select a subset of permutations from $M!$ permutations such that the minimum Hamming distance is increased in order to improve the BLER, at the expense of the communication data rate. We also answer Questions 2 by proposing efficient encoding schemes to map the incoming data symbols to corresponding waveforms under each subsets. 
	
	\section{Subset selection based on communication performance}\label{Sec-Comm}
	In this section, we focus on the communication performance and propose two methods of selecting a subset of waveforms such that the BLER can be improved.	
	
	\subsection{Permutation Subset with $d_{min}=2k$}\label{dmin_general}
	We first note that there exists no known method to find the maximum number of permutations with a given minimum Hamming distance, $d_{min}$, for general $M$ such that $3 < d_{min}< M$. In fact, the maximum size of the permutation subset for a given minimum Hamming distance is unknown for large $M$ \cite{52214.71}. Therefore, we first propose a block based approach where we separate $M$ frequency tones into blocks of $k$ tones such that $M/k$ remains an integer. As such, we define a new sequence of equally spaced blocks $[B_0, B_1,...,B_{M/k-1}]$ such that each block $B_j$ is mapped to a specific order of frequency tones given as $B_j=[f_{jk},f_{jk+1},...,f_{(j+1)k-1}]$. By limiting each block to be used only once, we can obtain $(M/k)!$ potential waveforms which correspond to all the permutations generated from the symbol set $\{1,2,...,M/k\}$. For example, in Fig. \ref{figure11} we illustrate two such frequency sequences under this subset for (a) $k=2$, and (b) $k=3$. As the figure clearly shows in each permutation the blocks are permuted while keeping the frequencies within the blocks fixed thus, increasing the minimum Hamming distance. More specifically, under this approach, when considering the minimum Hamming distance at least two blocks are different from each other with each block having $k$ frequency tones. Therefore, the minimum Hamming distance will increase to $2k$ while the number of total permutations in the selected subset, $\mathbf{\mathcal{S}}$, reduces to $(M/k)!$.
	\begin{figure}
		\begin{subfigure}{.5\textwidth}
			\centering
			\includegraphics[width=\textwidth]{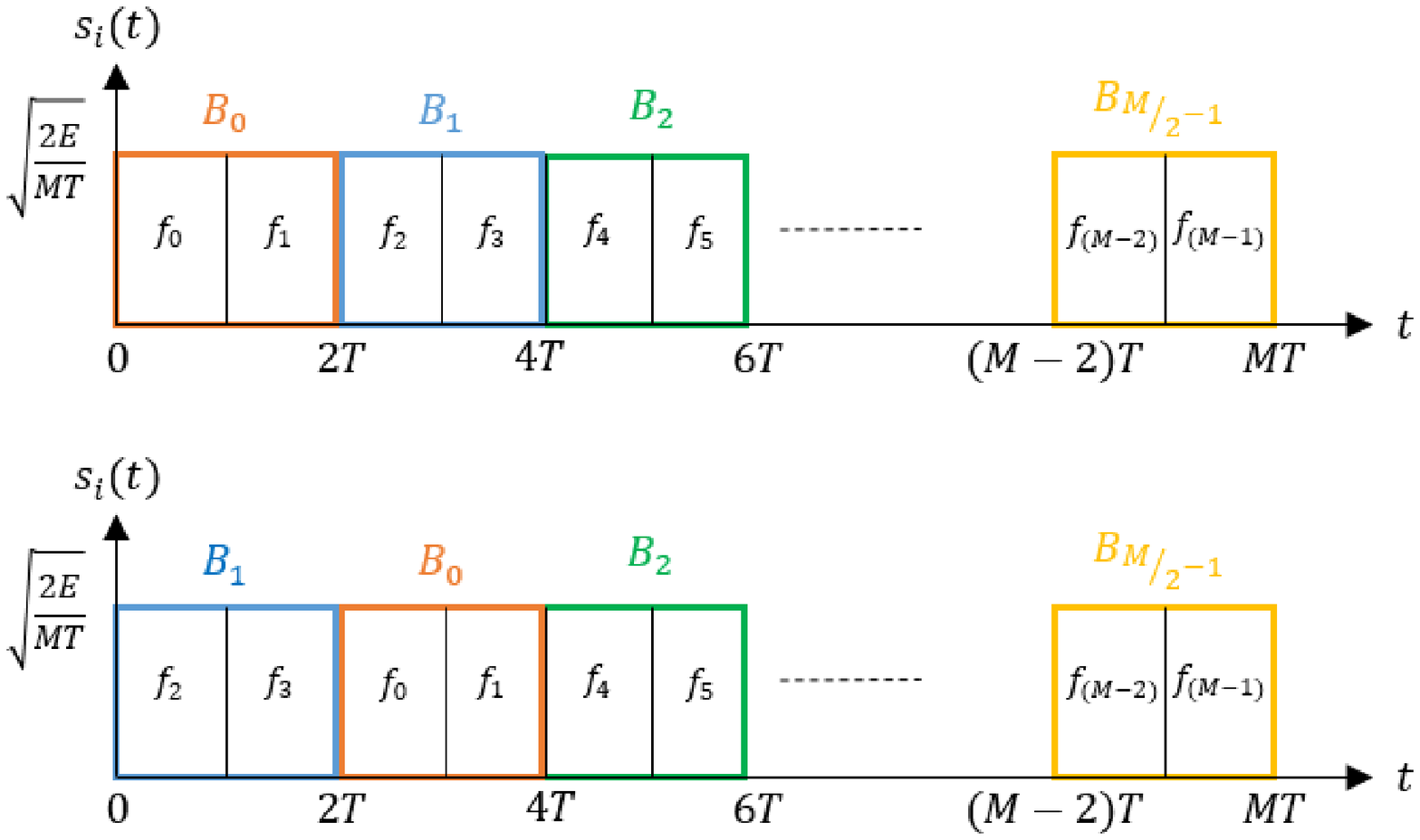}
			\caption{$k=2$}
		\end{subfigure}
		\begin{subfigure}{.5\textwidth}
			\centering
			\includegraphics[width=\textwidth]{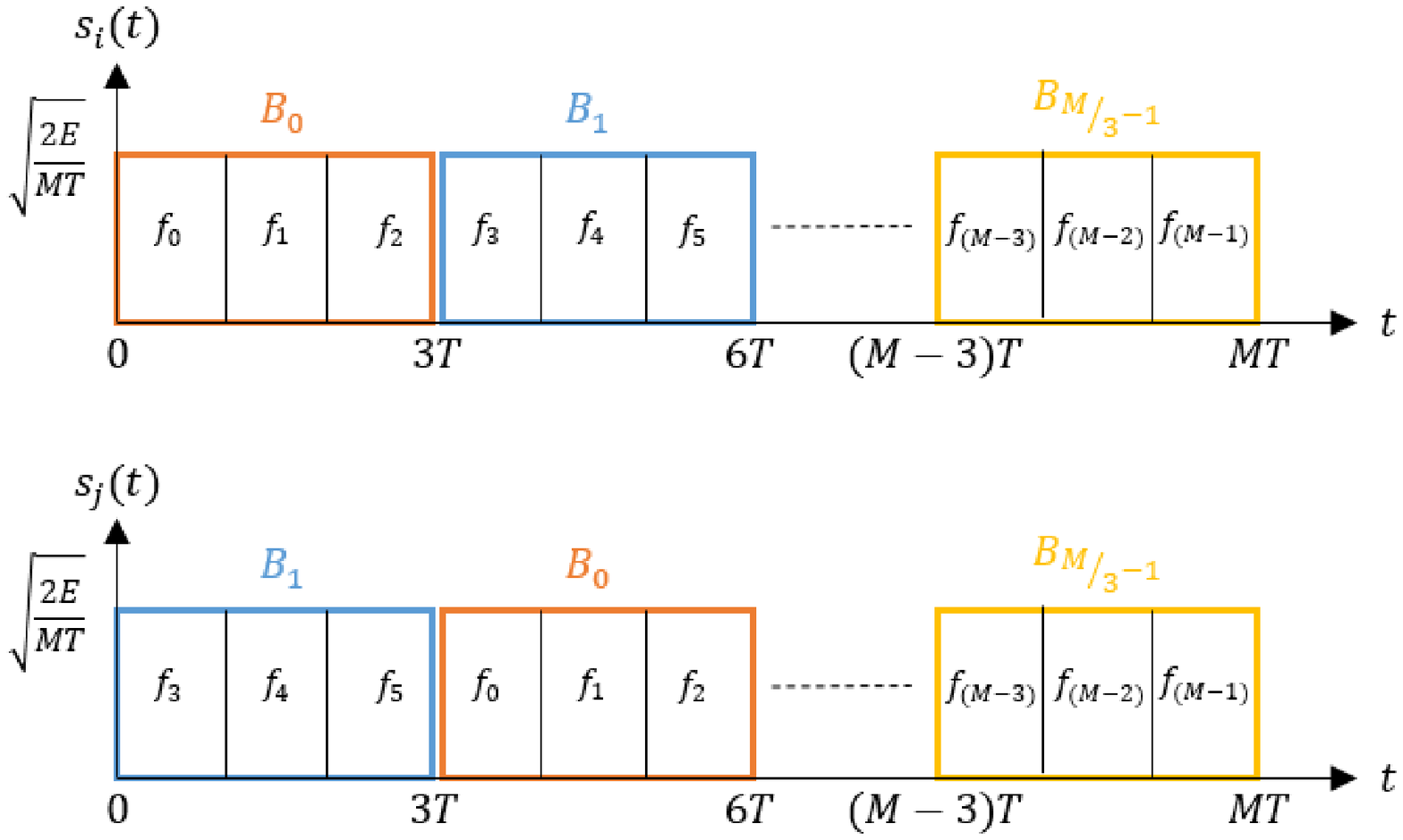}
			\caption{$k=3$}
		\end{subfigure}
		\caption{Two different frequency sequences under block based subset selection.}
		\label{figure11}
	\end{figure}
	We note that it is possible to find a larger subset with the same minimum Hamming distance of $d_{min}=2k$ by using a brute force exhaustive search of all possible subset combinations. Therefore, this block based approach limits the communication data rate that can be achieved for the same BLER. However, it provides a more feasible approach to subset selection for large $M$. In the following, we provide an efficient encoding scheme and an optimal receiver implementation with a worst case complexity of $O\big((M/k)^3\big)$ under the block based approach.
	
	\subsubsection{Encoding scheme}
	We note that due to the proposed block based approach, it is possible to map information symbols to their corresponding waveforms using the Lehmer code and factorial number system as follows. At the communication transmitter, the natural number of the incoming data symbol is first computed using the factorial number system. If the natural number of the incoming data symbol is given by $i\in\{0,1,...,(M/k)!-1\}$, we select the $i$-th  permutation, generated from the symbol set $\{1,2,...,M/k\}$ in lexicographic order. Since each block, $B_j$, is mapped to a specific order of frequency tones there exists a unique frequency mapping for each permutation generated from the symbol set $\{1,2,...,M/k\}$. Therefore, it is a straightforward conversion from the selected permutation to its corresponding frequency mapping to obtain the waveform for transmission. At the communication receiver, we first detect the permutation, generated from the symbol set $\{1,2,...,M/k\}$, corresponding to the received signal. If the $i$-th permutation in lexicographic order is detected, then the natural number of the received data symbol is taken as $i$.
	
	\subsubsection{Receiver implementation}
	In this section we present an efficient and optimal receiver implementation for the proposed block based approach. We first define a new correlation matrix
	\begin{align}\label{eq14}
	\hat{\mathbf{R}} = (\hat{r}_{uv}) \in \mathbb{R}^{M/k\times M/k},
	\end{align}
	where the $uv$-th element, $\hat{r}_{uv}$, which is given by 
	\begin{align}\label{eq15}
	\hat{r}_{uv} = \mathrm{Re} \biggl(\int_{(u-1)kT}^{ukT} \mathbf{h}^H\mathbf{r}(t)\hat{\psi}_v \big(t-(u-1)kT\big)\, dt\biggr),
	\end{align}
	denotes the correlation between the received signal,$\mathbf{h}^H\mathbf{r}(t)$, and the new basis function $\hat{\psi}_v \big(t-(u-1)kT\big)$ which can be expressed as, 
	\begin{align}\label{eq16}
	\hat{\psi}_v(t) = \sqrt{\dfrac{E}{kT}} \sum_{m=(v-1)k}^{vk-1} s_p(t-mT)\, \cos\bigg(2\pi f_m(t-mT)\bigg).
	\end{align}
	Then, the optimal receiver simplifies to selecting $M/k$ elements such that only one element is selected from each row and each column of $\hat{\mathbf{R}}$ such that the sum of the selected elements is maximized. Therefore, this optimization problem is an assignment problem and the Hungarian algorithm can be used to obtain the solution. Next, the output of the Hungarian algorithm can be converted to the corresponding permutation generated from the symbol set $\{1,2,...,M/k\}$ and the natural number related to the received data symbol can be decoded using the Lehmer code. As such the optimal receiver under this subset has a worst case complexity of $O\big((M/k)^3\big)$. 
	
	\subsubsection{Block error rate}
	Under this subset, we consider all $(M/k)!$ permutations generated from the symbol set $\{1,2,...,M/k\}$. As such, we consider the universal permutation set with regards to the blocks denoted by $B_i, \forall i=0,1,...,M/k-1$. When the frequency tones differ by $kl$ positions, we can see that there are $!l{{M/k}\choose{M/k - l}}$ pairwise error probabilities where $!b$ denotes the number of derangments of a $b$-element set and can be expressed as,
	\begin{align}\label{eq11}
	!b = b! \sum_{i=0}^{b}\dfrac{(-1)^i}{i!}.
	\end{align}
	Please note that a permutation of the sequence $\{1,2,...,b\}$ that has no fixed points is known as a derangement of a $b$-element\cite{Article03.1.2}. Therefore, considering the symmetry involved in the pairwise error probabilities, the union bounds given in \eqref{eq8.1} and \eqref{eq8.2} can be simplified to,
	\begin{align}\label{eq12.1}
	P_e^{UB} &= \sum_{l=2}^{M/k}!l {{M/k}\choose{M/k - l}} \mathcal{Q}\biggl(\sqrt{\dfrac{NEkl}{N_0 M}}\biggr), \\ 
	P_e^{UB} &= \sum_{l=2}^{M/k} {{M/k}\choose{M/k\!-\!l}} \dfrac{1}{\pi} \int_{0}^{\pi/2} \biggl[\prod_{n=1}^{N}\bigg(\dfrac{\alpha_{kl}\sin^2\theta}{\lambda_{n}+\alpha_{kl}\sin^2\theta}\bigg)\exp\bigg(\sum_{n=1}^{N}\dfrac{-K|(\mathbf{V}^H\Delta)_n|^2}{\lambda_{n}+\alpha_{kl}\sin^2\theta}\bigg)\biggr] d\theta,\label{eq12.2}
	\end{align}
	under the AWGN channel and the Rician fading channel, respectively. We also note that while the union bound provides a simple upper bound, it becomes loose for large $M$, specifically in the low SNR regime. Therefore, we also derive another approximation based on the nearest neighbor (NN) approximation \cite{Proakis_digitalCom} in which only the pairwise error probabilities with the nearest neighbors, which have a Hamming distance of $2k$, is considered as opposed to all the neighbors considered in the union bound. The resulting NN approximations can be written as,
	\begin{align}\label{eq13.1}
	P_e^{NN} &= \dfrac{(M/k)(M/k\!-\!1)}{2} \mathcal{Q}\biggl(\sqrt{\dfrac{2kNE}{N_0 M}}\biggr), \\
	P_e^{NN} &= \dfrac{(M/k)(M/k\!-\!1)}{2\pi} \int_{0}^{\pi/2} \biggl[\prod_{n=1}^{N}\bigg(\dfrac{\alpha_{2k}\sin^2\theta}{\lambda_{n}+\alpha_{2k}\sin^2\theta}\bigg)\exp\bigg(\sum_{n=1}^{N}\dfrac{-K|(\mathbf{V}^H\Delta)_n|^2}{\lambda_{n}+\alpha_{2k}\sin^2\theta}\bigg)\biggr] d\theta, \label{eq13.2}
	\end{align}
	which are more accurate than the union bounds when $M$ is large under the AWGN channel and the Rician fading channel, respectively.
		
	\subsection{Permutation Subset with $d_{min}=3$}\label{dmin_3}
	In the area of permutation arrays, it is well known that the maximum size of the permutation array with the minimum Hamming distance $d_{min}$ decreases significantly with increasing $d_{min}$ \cite{52214.71}. Therefore, while the BLER can be reduced by increasing $d_{min}$, the loss of communication data rate can become significant. As such, in this sub-section, we consider the special case of $d_{min} = 3$ which is the next best $d_{min}$ we can achieve after $d_{min} = 2$ without having a significant effect on the communication data rate. To introduce this subset selection method, first let us define some preliminary terms used with permutations. Let $\chi_i = [\chi_i(1), \chi_i(2), \ldots, \chi_i(M)]$ denote a random permutation of integers $1,2, \ldots, M$, with $\chi_i(m)$ denoting the $m$-th integer. Then,  
	\begin{itemize}
		\item the number of inversions in $\chi_i$ is defined as,
		\begin{align*}
		N(\chi_i) = \{(m,n):m<n \textrm{ and } \chi_i(m) > \chi_i (n) \}.
		\end{align*}
		\item the sign of permutation $\chi_i$ is positive $(+1)$ if $N(\chi_i) $ is even and negative $(-1)$ if  $N(\chi_i)$ is odd \cite{Perm_group}.
	\end{itemize}
	In \cite{52214.71}, it is shown that the alternative group consisting of the permutations with the same sign has $d_{min}=3$ and that the size of this subset is given by $M!/2$. Therefore, in this work we consider the alternative group that consists of the permutations which have a positive sign.
	
	\subsubsection{Encoding scheme}
	Due to the special structure of the alternative group, it is possible to design an efficient scheme that maps the incoming data symbols to the corresponding waveforms in the selected subset. We first provide the following Lemma.
	
	\begin{lemma}\label{Lemma1}
		According to the lexicographic order, the two adjacent permutations given by $2i$ and $(2i+1)$ have different signs.
	\end{lemma}
	\begin{proof}
		Given in Appendix \ref{appendixA}.
	\end{proof}

	According to Lemma \ref{Lemma1}, incoming data symbols can be mapped into permutations as follows. If the natural number	of the incoming data symbol is given by $i$, according to the factorial number system, we compute the sign of the $(2i)$-th permutation and the $(2i+1)$-th permutation in the lexicographic order. Then, the waveform corresponding to the permutation with the positive sign is selected for transmission. At the communication receiver, we first detect the permutation corresponding to the received signal. If the $i$-th permutation is detected, then the natural number of the received data symbol is taken as $\lfloor i/2\rfloor$. Therefore, even under this subset the mapping from incoming data symbols to corresponding waveforms can be implemented very efficiently without the use of a large look-up table.
	
	\subsubsection{Receiver implementation}	
	While an efficient optimal receiver implementation is not feasible under this subset, we note that the sub-optimal receiver in Algorithm \ref{Algorithm1} with $d\!=\!2$ results in close to optimal performance. This can be explained as follows. If the Hungarian algorithm outputs a permutation that does not belong to $\mathbf{\mathcal{S}}$, then it has a negative sign. As a result, all $M(M\!-\!1)/2$ neighbors of that permutation with the Hamming distance 2 have a positive sign and thus belong to $\mathbf{\mathcal{S}}$. The most common errors result in receiving one or two tones erroneously. Therefore, in the case where the output of the Hungarian algorithm in not within $\mathbf{\mathcal{S}}$, there is a high probability that one of its nearest neighbors with Hamming distance 2 is the highest correlated waveform within $\mathbf{\mathcal{S}}$. As a result, our proposed sub-optimal receiver obtains close to optimal BLER performance under this subset and this is further illustrated using simulations in Section \ref{Sec-Simu}.
	
	\subsubsection{Block error rate}
	As the minimum Hamming distance is now increased to $3$ compared to the universal set, it is possible to detect up to two errors and correct any single error. As such, this subset of permutations provides a coding scheme equivalent to single parity-bit coding in communication networks. Further, due to the symmetric structure in pairwise error probabilities in the alternative group, the union bounds given in \eqref{eq8.1} and \eqref{eq8.2} can be simplified to,
	\begin{align}\label{eq17.1}
	P_e^{UB} &= \sum_{l=3}^{M}  A_l \mathcal{Q}\biggl(\sqrt{\dfrac{NEl}{N_0 M}}\biggr), \\ 
	P_e^{UB} &= \sum_{l=3}^{M} \dfrac{A_l}{\pi} \int_{0}^{\pi/2} \biggl[\prod_{n=1}^{N}\bigg(\dfrac{\alpha_l\sin^2\theta}{\lambda_{n}+\alpha_l\sin^2\theta}\bigg)\exp\bigg(\sum_{n=1}^{N}\dfrac{-K|(\mathbf{V}^H\Delta)_n|^2}{\lambda_{n}+\alpha_l\sin^2\theta}\bigg)\biggr] d\theta,\label{eq17.2}
	\end{align}
	under the AWGN channel and the Rician fading channel, respectively. A special property of the alternative group is that for any permutation with a positive sign, all of its neighbors with Hamming distance $3$ also have a positive sign and vice versa. As such, we can derive the NN approximation by only considering the pairwise error probabilities corresponding to the nearest neighbors which have a Hamming distance of $3$ and obtain
	\begin{align}\label{eq17.3}
	P_e^{NN} &= \dfrac{M\!(M\!-\!1)\!(M\!-\!2)}{3} \mathcal{Q}\biggl(\sqrt{\dfrac{3NE}{N_0 M}}\biggr), \\
	P_e^{NN} &= \dfrac{M\!(M\!-\!1)\!(M\!-\!2)}{3\pi} \int_{0}^{\pi/2} \biggl[\prod_{n=1}^{N}\bigg(\dfrac{\alpha_3\sin^2\theta}{\lambda_{n}+\alpha_3\sin^2\theta}\bigg)\exp\bigg(\sum_{n=1}^{N}\dfrac{-K|(\mathbf{V}^H\Delta)_n|^2}{\lambda_{n}+\alpha_3\sin^2\theta}\bigg)\biggr] d\theta, \label{eq17.4}
	\end{align}
	as an alternative approximation which is more accurate than the union bound when $M$ is large under the AWGN channel and the Rician fading channel, respectively.

	\section{Subset selection based on radar performance}\label{Sec-Radar}	
	In this section, we consider the selection of a subset of permutations to improve the radar performance of the transmitted waveform. We focus on the ambiguity function (AF) which is an important analytical tool used in radar signal analysis. In radar sensing applications, the performance of the radar waveform is measured based on local and global accuracy. The local accuracy is defined as the measurement accuracy and the resolution, which can be measured in terms of the range resolution, the Doppler resolution and the efficiency of spectral usage, and depends on the properties of the mainlobe of the AF. The global accuracy is defined as the level of interference which can be measured by the sidelobe behavior of the AF \cite{0471663085}.
	In \cite{12967}, it is shown that the auto-correlation function, which determines the properties of the mainlobe, does not depend on the frequency order. As such, the local accuracy remains the same for all the permutations and it will not be impacted by selecting a subset of permutations. Therefore, in this work, we focus on improving the global accuracy of radar sensing via the selection of a subset of permutations such that the number of larger peaks in the sidelobe region of the AF is minimized.
	
	\subsection{Ambiguity Function}
	In radar sensing, the AF provides the matched filter output when the transmitted radar signal is received with a certain time delay and a Doppler shift. In other words, the AF provides a measure of dependence between the transmitted signal and its frequency and time shifted replica \cite{0471663085}. For the transmitted waveform $s_i(t)$, the complex AF can be expressed as,
	\begin{align}\label{eq23}
	\hat{A}(\tau,f_d) = \int_{-\infty}^{\infty} s_i^{*}(s)s_i(s-\tau)\,  \textrm{exp}(j2\pi f_d \tau) d\tau,
	\end{align}
	where $s_i^{*}(t)$ is the complex conjugate of $s_i(t)$, $\tau$ is the time delay relative to the output peak of the matched filter and $f_d$ is the Doppler mismatch. Then, the AF of the waveform is given by the magnitude of $\hat{A}(\tau,f_d)$. Without loss of generality, the signal energy, $E$, is normalized to unity. Next, we substitute \eqref{eq1} into \eqref{eq23} and do some mathematical manipulations to re-express  \eqref{eq23} as,
	\begin{align}\label{eq24}
	\hat{A}(\tau,f_d) = \dfrac{1}{M}\sum_{n=0}^{M-1} \textrm{exp}(j2\pi f_d\,nT) \biggl[\phi_{nn}(\tau,f_d) + \sum_{m=0,m\neq n}^{M-1}\phi_{nm}(\tau-(n-m)T,f_d)\biggr],
	\end{align}
	where $\phi_{nn}(\tau,f_d)$ and $\phi_{nm}(\tau,f_d)$ denote the auto-correlation function and the cross-correlation function of the $n$-th pulse of the transmitted waveform and the $m$-th pulse of its time and frequency shifted replica, respectively. After some mathematical manipulations these can be expressed as,
	\begin{align}\label{eq25}
	\phi_{nn}(\tau,f_d) &= \left\{\!\!
	\begin{array}{ll}
	\dfrac{(T\!-\!|\tau|)}{T}\dfrac{\sin(\pi f_d(T\!-\!|\tau|))}{\pi f_d(T\!-\!|\tau|)}\textrm{exp}(j\pi f_d(T\!+\!\tau)\!-\!j2\pi f_n^i \tau) \qquad \quad & \textrm{if } |\tau| \le T \\
	0 & \textrm{otherwise,}
	\end{array}
	\right.\\
	\phi_{nm}(\tau,f_d) &= \left\{\!\!
	\begin{array}{ll}
	\dfrac{(T\!-\!|\tau|)}{T}\dfrac{\sin(\pi \alpha (T\!-\!|\tau|))}{\pi \alpha(T\!-\!|\tau|)}\textrm{exp}(-j\pi \alpha (T\!+\!\tau)\!-\!j2\pi f_m^i \tau) \qquad \quad & \textrm{if } |\tau| \le T \\
	0 & \textrm{otherwise,}
	\end{array}
	\right.\label{eq26}
	\end{align}
	with $\alpha = f_n^i-f_m^i-f_d$. The matched filter output in the presence of zero delay is given by the zero-delay response $\hat{A}(0,f_d)$. When $\tau=0$, all the terms related to cross-correlation become zero as $|(n-m)T|>T, \forall m\neq n$. Therefore, the zero-delay response only depends on the auto-correlation function. As such, the zero-delay response can be simplified using \cite[Eq. (26)]{12967} as,
	\begin{align}\label{eq27}
	\hat{A}(0,f_d)= \dfrac{\sin(\pi Mf_dT)}{\pi Mf_dT}\textrm{exp}(j\pi Mf_dT).
	\end{align}
	From \eqref{eq27}, we can see that the zero-delay response of the AF does not depend on the frequency order of the permutation. As such, the Doppler resolution remains the same for all the permutations. 
	
	Next, we consider the matched filter output when there is no Doppler mismatch, which is given by the zero-Doppler response $\hat{A}(\tau,0)$. We note that the delay resolution only depends on the properties of the mainlobe of the AF which in turn relates to the auto-correlation function \cite{12967}. Therefore, in order to analyze the delay resolution of the proposed waveform, we first consider the part of $\hat{A}(\tau,0)$ determined by the auto-correlation function, $\phi_{nn}(\tau,f_d)$, and denote it by $\hat{A}^{'}(\tau,0)$. Again using \cite[Eq. (26)]{12967}, we can express this as, 
	\begin{align}\label{eq28}
	\hat{A}^{'}(\tau,0)= \left\{\!\!
	\begin{array}{ll}
	\dfrac{(T\!-\!|\tau|)}{T}\dfrac{\sin(\pi Mq\tau/T)}{M \sin(\pi q\tau/T)} \textrm{exp}\bigg(\!\!-j\pi\bigg(\!2f_c\tau\!+\!(M\!-\!1)\dfrac{\tau}{T}\bigg)\bigg) \qquad \quad & \textrm{if } |\tau| \le T \\
	0 & \textrm{otherwise,}
	\end{array}
	\right.
	\end{align}
	where $f_c$ denotes the carrier frequency of the transmitted waveform. From \eqref{eq28}, we can see that $\hat{A}^{'}(\tau,0)$ does not depend on the frequency order of the permutation. As a result, the location of the first zero in the zero-Doppler response does not change based on the permutation and the delay resolution remains the same for all the permutations. We also note that the inverse power spectrum is equivalent to the auto-correlation function \cite{0471663085}. Therefore, the power spectrum and the efficiency of spectral usage remains constant for all the permutations.    
	
	Let us now focus on the sidelobe behavior determined by the cross-correlation terms. From \eqref{eq26} and the properties of the sinc(.) function it follows that the peak of each $\phi_{nm}(\tau-(n-m)T,f_d)$ term occurs at $\tau_p=(n-m)T$ and $f_p=f_n^i-f_m^i$. Without loss of generality, let us consider the positive delay axis with $n>m$ and $L=n-m$. We can show that the peak of each $\phi_{nm}$ term occurs at $\tau_p=LT$ and $f_p = q\Delta_{L,m}/T$, where $\Delta_{L,m}$ denotes the number of frequency spacings $\Delta_f$ between the frequency tones in $m$-th and $(m\!+\!L)$-th positions in the transmitted permutation. As such, we can see that the behavior of the sidelobes are determined by the frequency order of the permutation. Given that $\Delta_{L,m}$ is an integer, peak points along a specific delay axis can be obtained by setting $f_d=\pm kq/T$, where $k$ is an positive integer. First, we will analyze the behavior of a single cross-correlation term in \eqref{eq24} at these peak points along the delay axis $\tau=\tau_p$ and write
	\begin{align}\label{eq29}
	\phi_{nm}(\tau_p-(n-m)T,f_d) = \dfrac{\sin(\pi q(\Delta_{L,m}\pm k))}{\pi q(\Delta_{L,m}\pm k)} \textrm{exp}\bigg(\!-j\pi(\Delta_{L,m}\pm k)\bigg).
	\end{align}
	Based on the behavior of the sinc(.) function, we can see that in the sidelobe region, $\phi_{nm}(\tau_p-(n-m)T,f_d)$ has a peak value of $1$ when $k=|\Delta_{L,m}|$ and for all other integer coordinate pairs of normalized delay and Doppler shift, $(\tau/T,f_dT)$ the value is $0$. 
	
	Based on this, we next analyze the combined behavior of all the cross-correlation terms in \eqref{eq24}. Let us first consider the case where $\Delta_{L,m}$ is unique for each $L$. In this case, each sidelobe term $\phi_{nm}$ has a unique normalized Doppler shift $f_dT = q\Delta_{L,m}$ related to the maximum point at the fixed normalized delay $\tau/T=L$. Therefore, the peaks of the $\phi_{nm}$ terms will be spaced by integral values in $\tau/T$ and $f_dT$ resulting in no superposition among the peaks. Such permutations are known as Costas permutations. In \cite{12967}, it is shown that for Costas permutations, on average the height of peaks in the sidelobe region is limited to $1/M$, while the random phasing of the sidelobes can produce isolated peaks in the order of $6$ dB over the $1/M$ value. Due to this property, Costas permutations are considered to provide higher global accuracy in radar sensing applications. However, only a very few number of Costas codes are available for larger $M$ and for some $M$ it is shown that there exists no Costas permutations \cite{5013-2}. While only a single permutation is needed in radar applications, in joint radar and communication systems, we require a large number of permutations to achieve a reasonable data rate. Therefore, in this work, we generalize the behavior of Costas permutations by relaxing the strict requirement of a unique $\Delta_{L,m}$ for each $L$ and propose a radar based permutation subset as follows. 
		
	\subsection{Radar based Permutation Subset}\label{costas_based}
	From \eqref{eq29}, we can see that if the same $\Delta_{L,m}$ is repeated several times for a given $L$, then several $\phi_{nm}$ terms will have their peak at the same $(\tau,f_d)$ coordinate. To evaluate the peak sidelobe height resulting from the superposition of several $\phi_{nm}$ terms, we define a difference triangle, as in \cite{12967}, where each row represents one of the possible values, $L=1,...,M-1$, and each column represents one of the possible values, $m = 0,1,...,M-1-L$, and the value in $m$-th column in $L$-th row is $\Delta_{L,m}$. Next, we can compute the maximum number of repeats in the difference triangle for any row and denote this as $g_L$. Therefore, we have $g_L$ cross-correlation terms overlapping at their peak coordinate resulting in a maximum sidelobe height of $g_L/M$ in the delay axis $\tau=LT$. As a result, over the entire delay axis, on average the height of the sidelobe peaks would be limited to ${\underset{L \in \{1,...,M-1\}} {\textrm{max}}} (g_L/M)$. Therefore, we can see that larger the value of ${\underset{L \in \{1,...,M-1\}} {\textrm{max}}} (g_L/M)$, larger the peak sidelobe height in the sidelobe region. Thus, we use the maximum number of repeats in the difference triangle as a measure of global accuracy in radar sensing. If the maximum number of repeats is small for a given permutation, that results in a high global accuracy for that corresponding waveform. Therefore, we first order the permutations in ascending order with respect to the maximum number of repeats in the difference triangle. Then, the permutations are selected starting from the smallest number of repeats and continued until we reach the required number of permutations needed to achieve a sufficient data rate for communication.
	
	However, under this selection method, there is no special structure in the selected permutation subset. Therefore, the mapping between incoming data symbols and their respective permutations requires a look-up table. Similarly, at the communication receiver, an efficient optimal receiver implementation is not feasible. As such, either the IP based optimal receiver in \eqref{eq20} or the sub-optimal receiver in Algorithm~\ref{Algorithm1} is required. 
	
	\subsection{Discussion}\label{joint_subset}
	So far we have selected subsets to either improve communication performance or radar performance. In the following, we focus on two aspects of performance improvements namely, the communication error rate and the radar sidelobe height, and discuss their joint impact on subset selection.
	
	We first note that due to the symmetric structure of the universal permutation set, for any given permutation there exists a reverse permutation whose AF is an exact mirror image of the original permutation. As such, the peak sidelobe height of these two permutations remains the same. We also note that the reverse permutation can be obtained by interchanging $\lfloor M/2\rfloor$ disjoint positions in the original permutation. In \cite{Perm_group}, it is shown that any interchange of two positions changes the sign of a permutation. Therefore, when $\lfloor M/2\rfloor$ is odd, reverse permutation has the opposite sign of the original permutation. As such, when we select half of the permutations under the $d_{min}=3$ subset, they are selected such that the number of permutations with the same peak sidelobe height is halved. Thus, the radar performance of this subset is equivalent to the universal set. On the other hand, when $\lfloor M/2\rfloor$ is even, we cannot guarantee that the permutations are selected such that the number of permutations with the same peak sidelobe height is halved. However, we can expect this impact on the global accuracy of radar sensing to be small on the basis that no specific pattern is considered during the subset selection.
	
	Next, we discuss how the communication error rate and the radar sidelobe height can be combined together to select a common subset such that both the communication performance and radar performance can be improved. While the development of a systematic approach to select a common subset is beyond the scope of this paper, here we present an example which contains a potential method to select a subset of waveforms that can improve both communication and radar performances. We note that the communication performance depends on the minimum Hamming distance in the selected subset while the radar performance depends on the number of repeats in the difference triangle for the ordered frequency sequence related to a given permutation. Based on this observation, we first select the permutation subset based on the communication perspective. Then, we map a difference frequency sequence to the number sequence $\{1,2,...,M\}$ to obtain a set of waveforms with a lower number of repeats in their respective difference triangles. 
	We note that while this method can improve the radar performance of the pre-selected permutation subset as illustrated in the following example, the development of a systematic approach to identify the best mapping is beyond the scope of this paper. 
	
	\vspace{0.25cm}
	\noindent
	\textbf{Example 1}: We set $M\!=\!8$ and select $M!/2\!=\!20160$ permutations from a total of $M!\!=\!40320$ permutations by considering the $d_{min}=3$ subset. Table \ref{table1} provides the number of permutations in this subset for each possible maximum number of repeats in the difference triangle under two different frequency mappings. Under the first mapping, we consider the basic frequency mapping of $[1,2,3,4,5,6,7,8] \rightarrow [f_1,f_2,f_3,f_4,f_5,f_6,f_7,f_8]$ and under the second mapping, we consider the mapping of $[1,2,3,4,5,6,7,8] \rightarrow [f_1,f_2,f_3,f_4,f_5,f_6,f_8,f_7]$. From the table, we can see that compared to the basic frequency mapping, the second frequency mapping removes the two waveforms with $6$ repeats (the maximum number) and also reduces the number of waveforms with $5,4,3$ and $1$ repeats by increasing the number of waveforms with $0$ and $2$ repeats. Further, we can see that the number of Costas permutations, which have the best radar performance, within the selected subset has increased by $124$ (from $160$ to $284$). Therefore, the global accuracy of the radar sensing is higher under the second frequency mapping while maintaining the same communication error rate and data rate.
	\begin{table}[H]
		\centering
		\caption{Number of permutations with maximum number of repeats in the difference triangle}
		\begin{tabular}
			{ |C{5cm}|C{5cm}|C{5cm}| }
			\hline
			Maximum number of repeats & No. of permutations under mapping 1 & No. of permutations under mapping 2 \\ \hline 
			6 & 2 & 0  \\ \hline
			5 & 16 & 8  \\ \hline
			4 & 134 & 130  \\ \hline		
			3 & 1212 & 1116 \\ \hline
			2 & 7100 & 7520 \\ \hline
			1 & 11536 & 11102  \\ \hline
			0 & 160 & 284 \\ \hline	
		\end{tabular}
		\label{table1}
	\end{table} 
	
	\section{Numerical Examples}\label{Sec-Simu}
	In this section, we provide numerical examples demonstrating the performance of different permutation subsets and their comparison against existing approaches. 
	
	While the analysis of the correlated fading channels presented in Section \ref{BLER} remains true for all $\mathbf{C}_u$, in the following, we simply consider the exponential correlation model \cite{2107.14396}. As such, the $(i,j)$-th entry of $\mathbf{C}_u$ is expressed as $\mathbf{C}_u(i,j) = \rho^{|i-j|}$ for a given correlation coefficient $\rho \in \{0,1\}$. We set the energy of the waveform $E = 1$ and the frequency separation $\Delta f = 1/T$ with $T = 1$, to maintain the orthogonality between frequencies. 
		
	\begin{figure}
		\centering
		\includegraphics[width=0.6\textwidth]{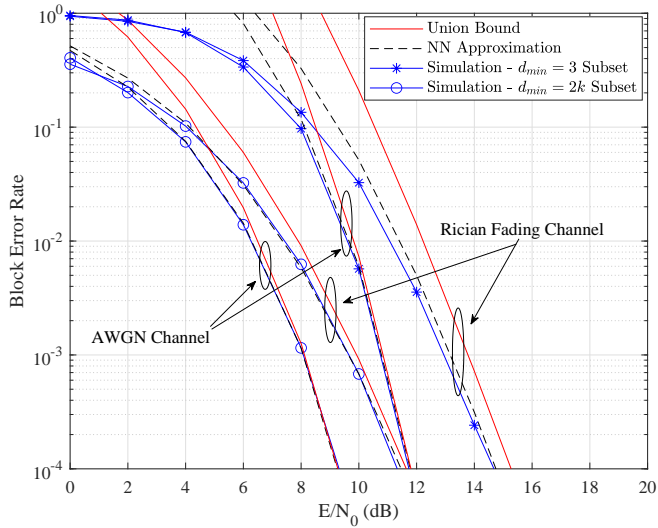}
		\captionsetup{justification=centering}
		\caption{The BLER versus the average received SNR under Hamming distance based permutation subsets with $M=8, N=4, K=2$ and $\rho=0.5$.}
		\label{figure6}
	\end{figure}
	Fig. \ref{figure6} plots the BLER versus the received signal-to-noise-ratio (SNR) for a communication receiver with $N=4$ antennas. We set $M=8$ frequency tones and plot the BLER performance for both the AWGN channel and a correlated fading channel with Rician factor $K=2$ and correlation coefficient $\rho=0.5$. We consider the $d_{min}=3$ subset and the $d_{min}=2k$ subset with $k=2$ and select $M!/2=20160$ and $(M/2)!=24$ waveforms, respectively. The union bounds in \eqref{eq12.1}, \eqref{eq12.2}, \eqref{eq17.1}, \eqref{eq17.2} and the NN approximations in \eqref{eq13.1}, \eqref{eq13.2}, \eqref{eq17.3}, \eqref{eq17.4} are used to generate te analytical approximations. From the plot we can observe that under the AWGN channel, the union bound accurately follows the simulation results in the high SNR regime, while the NN approximation is more accurate in the low SNR regime. However, under the Rician fading channel, the NN approximation is more accurate compared to the union bound even in the high SNR regime.
			
	Fig. \ref{figure3} plots the BLER versus the received SNR for an AWGN channel under the proposed permutation subsets and compares them against the universal set considered in \cite{3172111}. We set $N=2$ and $M=6$. Under the universal set, we consider all $M!=720$ waveforms. Under the $d_{min}=3$ subset, we select $M!/2=360$ waveforms such that each selected permutation has a positive sign. Under the $d_{min}=2k$ subset, we set $k=2$ thus obtaining a minimum Hamming distance of $4$ and select $(M/2)!=6$ waveforms. From the plot, we can observe that at $10$ dB we achieve a $4$-fold and more than $10$-fold reduction in the BLER compared to \cite{3172111} when $d_{min}$ is increased to $3$ and $4$, respectively. As expected the BLER reduces with increasing minimum Hamming distance, but at the expense of data rate. As the size of the subset decreases we have fewer waveforms to transmit data, thus reducing the resulting data rate. More specifically, the data rate drops by around $10\%$ and $70\%$ when we increase $d_{min}$ from $2$ to $3$ and from $3$ to $4$. As such, there is a clear trade-off between the data rate and the BLER and it is important to select a proper subset based on the communication requirements. 
	\begin{figure}
		\centering
		\includegraphics[width=0.6\textwidth]{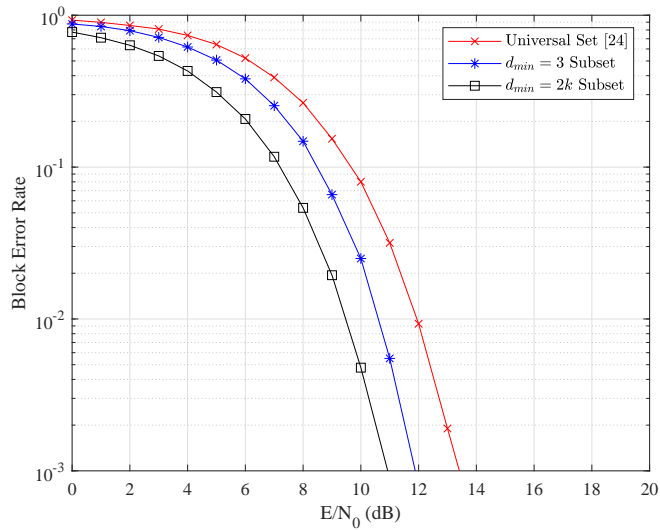}
		\captionsetup{justification=centering}
		\caption{The BLER versus the average received SNR under different Hamming distance based permutation subsets with $M = 6$ and $N = 2$.}
		\label{figure3}
	\end{figure}	

	Fig. \ref{figure4} plots the BLER versus the received SNR for a correlated fading channel with Rician factor $K=4$ and correlation coefficient $\rho=0.5$ under the proposed $d_{min}=3$ subset and the radar based subset in Section \ref{costas_based}. We set $N=4$ and $M=6$. We consider both the optimal receiver implementation using \eqref{eq20}, which is implemented via Matlab mixed-integer linear programming (MILP) toolbox, and the sub-optimal receiver implementation proposed in Algorithm \ref{Algorithm1}. Under the $d_{min}=3$ subset, $M!/2=360$ waveforms were selected such that each corresponding permutation has a positive sign. Under the radar based subset, we select $360$ waveforms to obtain a similar data rate such that the maximum number of repeats in the difference triangle for any selected permutation is less than or equal to one. 	
	From the plot, we can observe that the sub-optimal receiver in Algorithm \ref{Algorithm1} provides almost identical BLER results compared to the optimal receiver under the $d_{min}=3$ subset. This can be explained by the special structure of this subset, i.e., for a permutation not included in the selected subset, all Hamming distance $2$ neighbors belong to the selected subset. On the other hand, under the radar based subset, we can see that the BLER of the sub-optimal receiver is slightly higher than the optimal receiver, specially in the low SNR regime while in the high SNR regime it provides almost identical BLER results compared to the optimal receiver. Further, we can observe that at $10$ dB we achieve a $4$-fold reduction in the BLER when the $d_{min}=3$ subset is compared against the radar based subset. As such, we have sacrificed the communication performance in order to gain a better radar performance. 
	\begin{figure}
		\centering
		\includegraphics[width=0.6\textwidth]{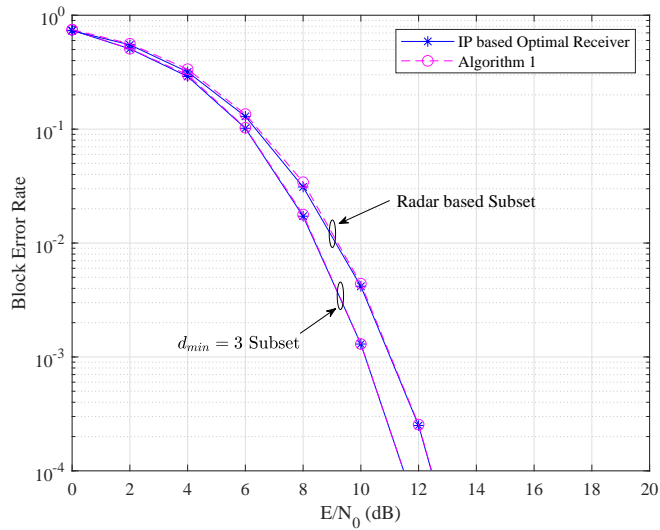}
		\captionsetup{justification=centering}
		\caption{The BLER versus the average received SNR under $d_{min}=3$ and radar based subsets with $M=6, N=4, K=4$ and $\rho=0.5$.}
		\label{figure4}
	\end{figure}
	
	Next, we focus on the radar performance of the transmitted waveform. We consider the same setup as in Fig. \ref{figure4} and compare the height distribution of the peak sidelobe under the proposed $d_{min}=3$ subset, the radar based subset and the universal set considered in \cite{3172111}. Fig. \ref{figure5} plots the histogram of the peak sidelobe height distribution when the mainlobe height is normalized to one. For the universal set, peak sidelobe heights vary from $0.2$ to $0.85$ with the mean peak sidelobe height of $0.3791$. Since, $\lfloor M/2\rfloor$ is odd, under the $d_{min}=3$ subset, the waveforms are selected such that the number of permutations with the same peak sidelobe height is halved. As a result, under the $d_{min}=3$ subset, the mean peak sidelobe height remains at $0.3791$. Therefore, there is no change in the radar global accuracy by selecting $360$ waveforms under the $d_{min}=3$ subset. On the other hand, under the radar based subset, the larger peak sidelobe heights have been removed and the highest peak sidelobe height is limited to $0.45$ with the mean peak sidelobe height of $0.3236$. This clearly indicates that the radar based subset achieves better global accuracy in radar sensing in comparison to $d_{min}=3$ subset and the universal set.
	\begin{figure}
		\centering
		\includegraphics[width=0.6\textwidth]{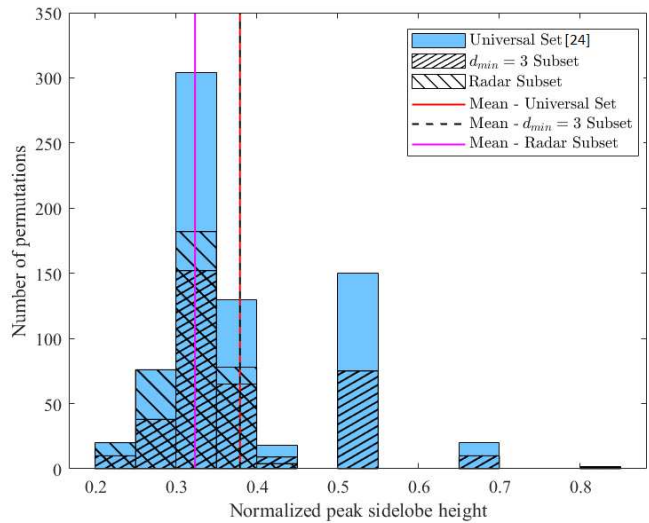}
		\captionsetup{justification=centering}
		\caption{Normalized peak sidelobe height distribution under $d_{min}=3$ and radar based subsets with $360$ waveforms and $M=6$.}
		\label{figure5}
	\end{figure}

	 Fig. \ref{figure7} and Fig. \ref{figure8} plot the AF of the transmitted waveform under two permutations for $M = 8$ frequency tones. Subplots (a) and (b) present the three-dimensional surface plot and the contour plot of the magnitude of $\hat{A}(\tau,f_d)$ versus the normalized delay $\tau/T$ and normalized Doppler $f_dT$, respectively. The maximum value of the AF is one due to normalized the signal energy. As expected, the maximum value occurs at $\tau=f_d=0$ and the behavior of the two waveforms is similar around the origin indicating that the local accuracy for all the permutations remains the same. On the other hand, it can be observed that the broader structure of the AF changes significantly based on the permutation corresponding to the transmitted waveform. In Fig. \ref{figure7}, we consider the permutation $[f_1,f_2,f_3,f_4,f_5,f_6,f_7,f_8]$ which has $6$ repeats in each row of the difference triangle. As such, we can see that there are peaks of height $0.875$ in the sidelobe region, which reduces the global accuracy of radar sensing. On the other hand, in Fig. \ref{figure8}, we consider the permutation $[f_1,f_2,f_4,f_5,f_8,f_7,f_6,f_3]$ which only has a maximum of one repeat in the difference triangle. As such, the height of peaks in the sidelobe region is limited to $0.25$ thus improving the global accuracy agreeing with out discussion in Section \ref{Sec-Radar}.
	\begin{figure}
		\begin{subfigure}{.475\textwidth}
			\centering
			\includegraphics[width=\textwidth]{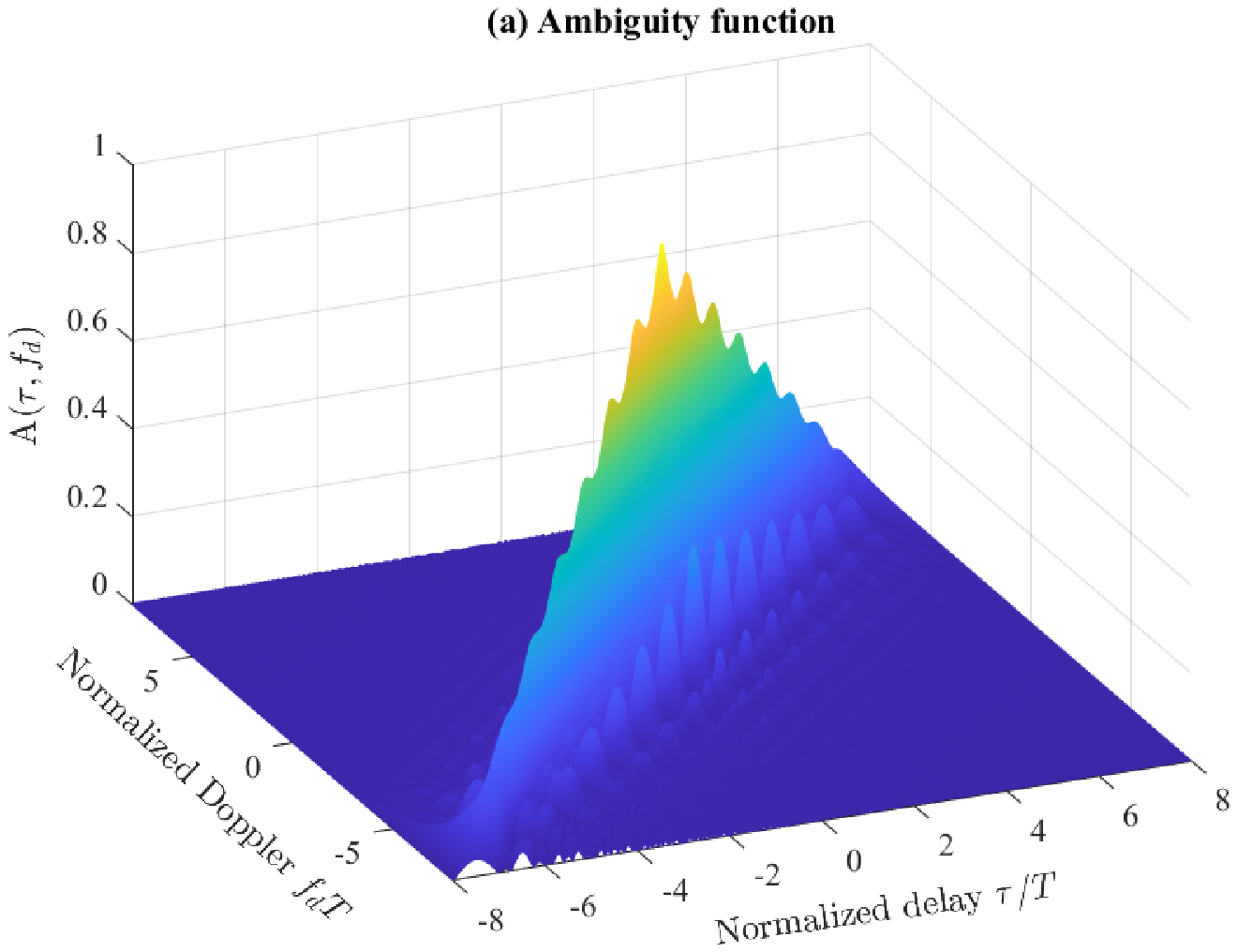}
		\end{subfigure}
		\begin{subfigure}{.475\textwidth}
			\centering
			\includegraphics[width=\textwidth]{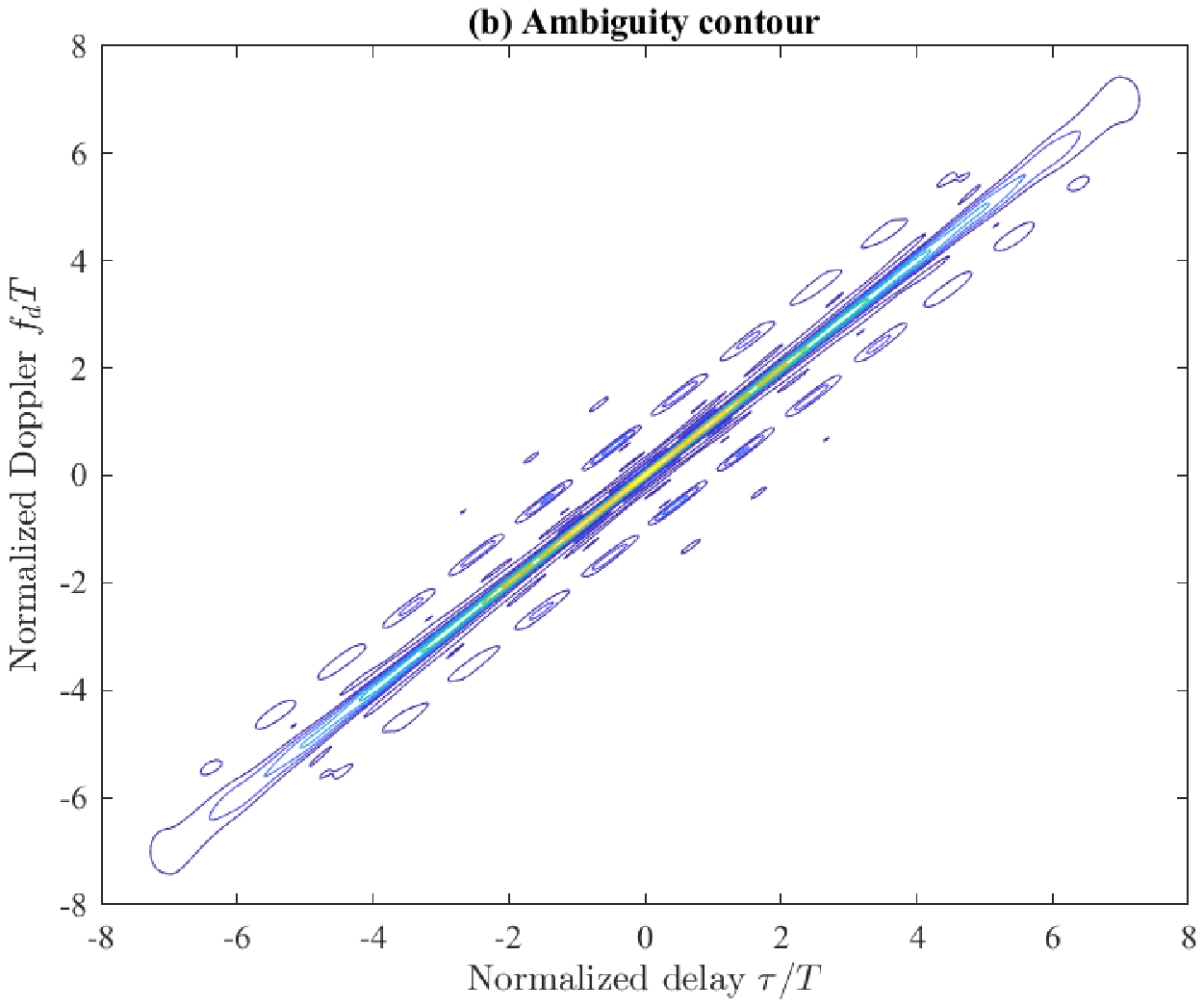}
		\end{subfigure}
		\caption{Ambiguity function of the waveform related to permutation $[f_1,f_2,f_3,f_4,f_5,f_6,f_7,f_8]$.}
		\label{figure7}
	\end{figure} 
	\begin{figure}
		\begin{subfigure}{.475\textwidth}
			\centering
			\includegraphics[width=\textwidth]{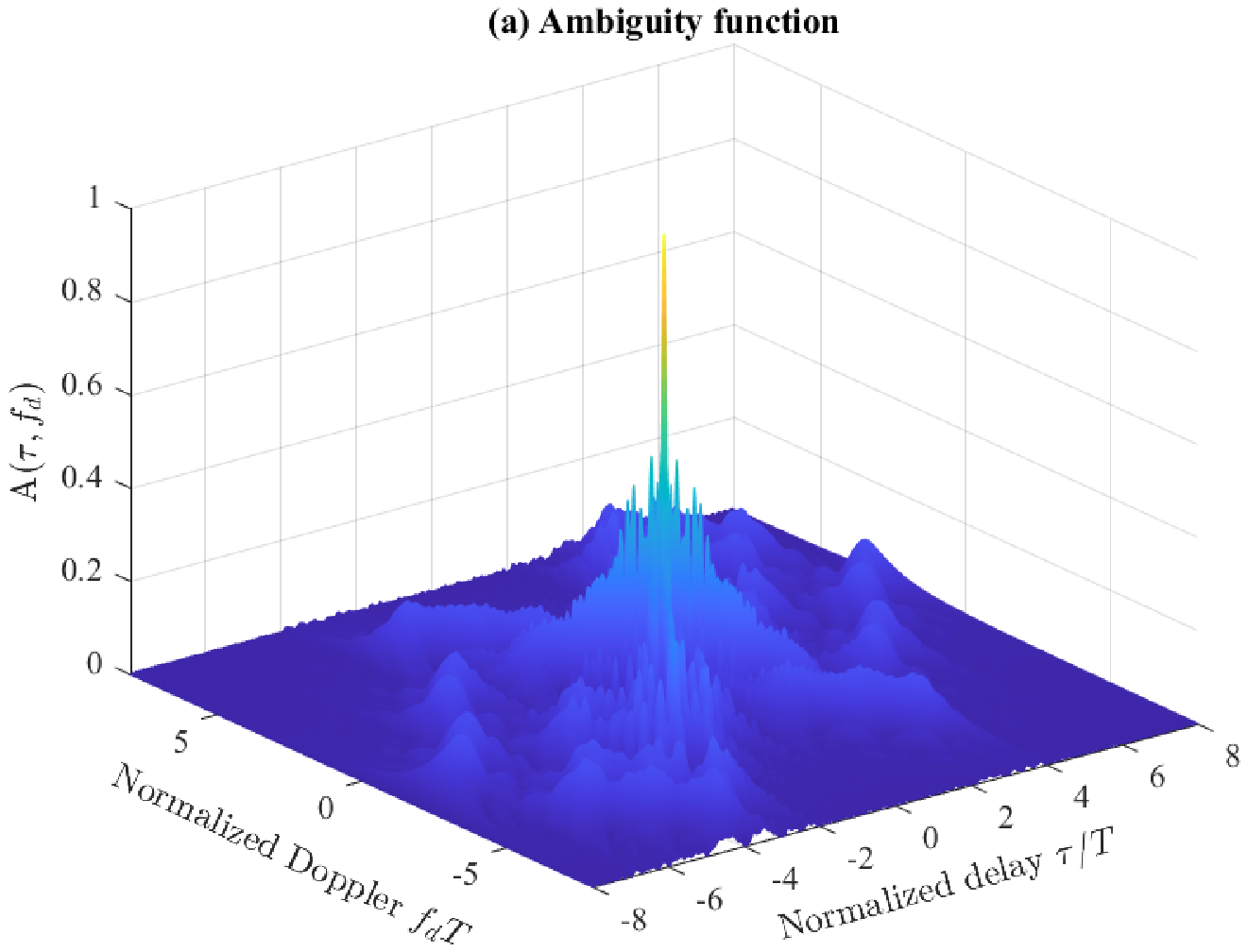}
		\end{subfigure}
		\begin{subfigure}{.475\textwidth}
			\centering
			\includegraphics[width=\textwidth]{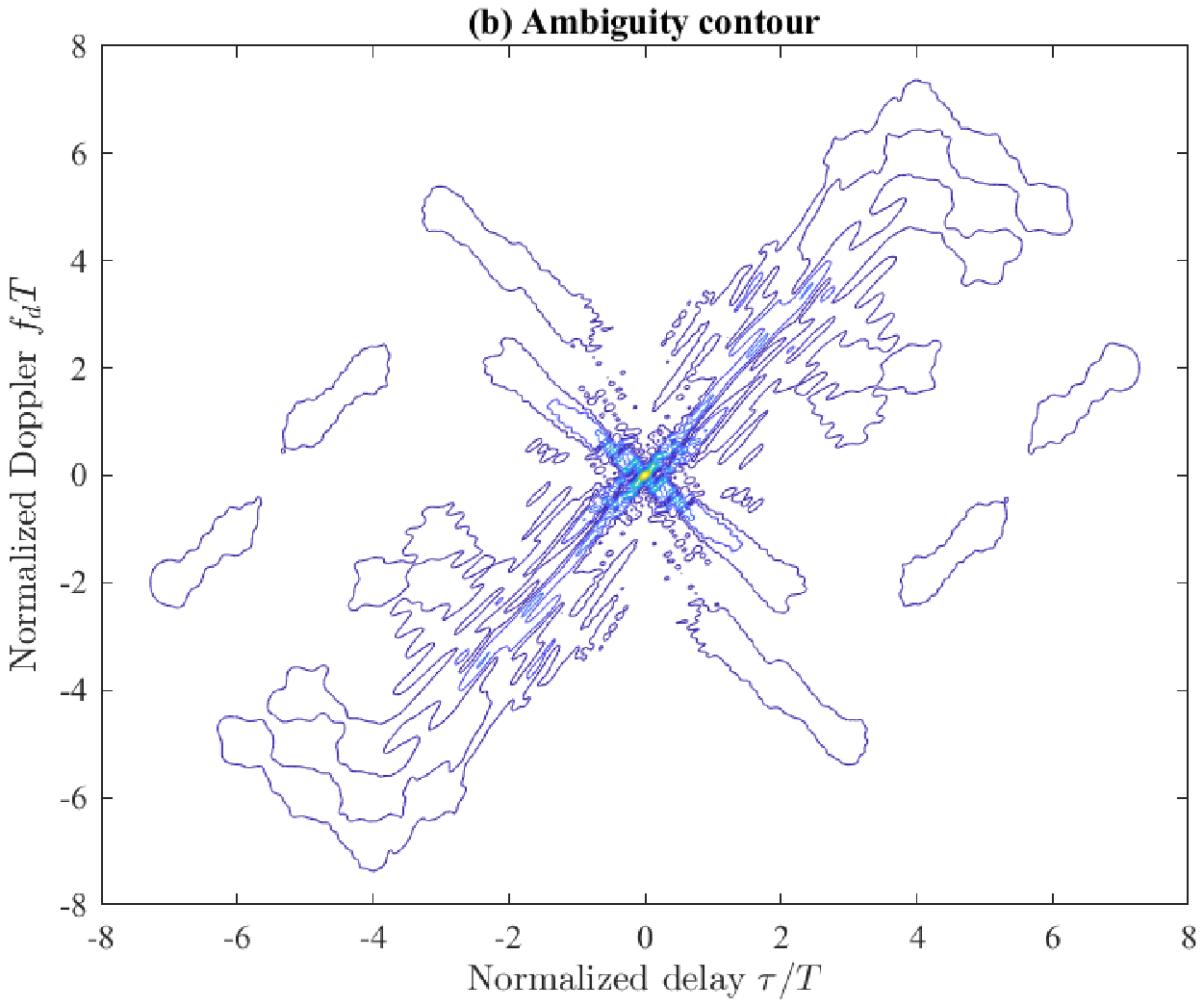}
		\end{subfigure}
		\caption{Ambiguity function of the waveform related to permutation $[f_1,f_2,f_4,f_5,f_8,f_7,f_6,f_3]$.}
		\label{figure8}
	\end{figure} 

	Next, we focus on the distribution of the permutations with respect to the maximum number of repeats in the difference triangle. Fig. \ref{figure10} plots the histogram of the maximum number of repeats in the difference triangle for the universal set with $M=7$ frequency tones. From the figure, we can see that there are only $200$ Costas permutations out of $M!=5040$. In order to improve the data rate, we can select up to $3262$ permutations with a maximum of $1$ repeat which increases the date rate by around $50\%$. As such, in order to increase the data rate further we need to go for $2$ repeats, $3$ repeats and so on. However, the consideration of permutations with a larger number of repeats reduces the radar performance due to larger sidelobe heights. Therefore, both the communication data rate and the radar performance need to be considered when selecting a proper subset size. 
	\begin{figure}
		\centering
		\includegraphics[width=0.6\textwidth]{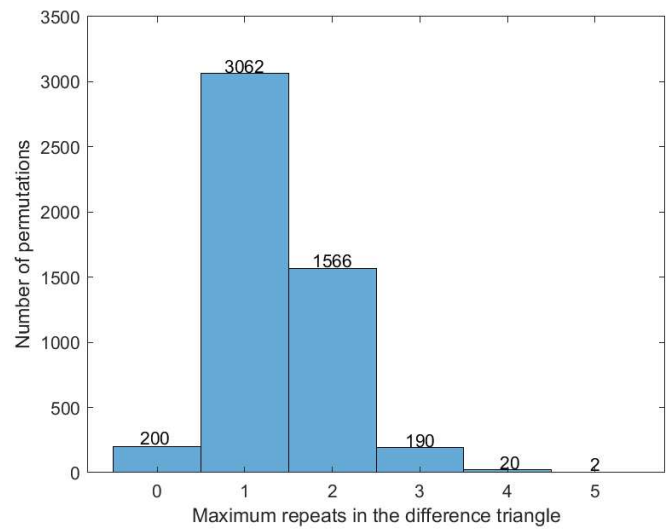}
		\captionsetup{justification=centering}
		\caption{Distribution of the maximum number of repeats in the difference triangle under the universal set with $M=7$.}
		\label{figure10}
	\end{figure}
	
	\section{Conclusion}\label{Sec-Conc}
	The waveform design of joint radar and communication systems is investigated and a new subset selection process is presented for the random stepped frequency permutation waveform in order to improve the communication and radar performances. First, we showed that the existing receiver implementation fails when a subset of permutations is selected from the universal set. As such, we proposed an optimal communication receiver based on integer programming to handle any subset of permutations followed by a more efficient sub-optimal receiver based on the Hungarian algorithm. Next, we analyzed the block error rate at the communication receiver under the optimum maximum likelihood detection, and proposed two methods to select a permutation subset with a given minimum Hamming distance such that the communication error rate can be improved. We also proposed an efficient encoding scheme to map the incoming data symbols to their corresponding permutations under these subsets. From the radar perspective, we analyzed the ambiguity function and showed that local accuracy does not depend on the subset selection. Finally, we extended the properties of Costas arrays and proposed another subset selection method to reduce the maximum sidelobe height such that the global accuracy of radar sensing can be improved.
	
	In this work, we introduced the remapping of the frequency tones to the symbol set used to generate permutations as a potential method to improve both the communication and radar performances of the selected permutation subset. One desirable future extension would be the development of a systematic process to obtain the best frequency mapping. We also note that in communication networks, the information generally takes the form of a bit stream and the number of data symbols generally take a power of two. As such, an extension to the bit error rate analysis by selecting a subset with a power of two permutations would formulate another challenging problem.
	
	\begin{appendices}		
		\section{Proof of Lemma \ref{Lemma1}}\label{appendixA}		
		Let us consider the $r$-th permutation in the lexicographic order, given by $\chi_r$, where $r \in \{0,1,...,M!-1\}$. We note that the number of inversions in the $r$-th permutation can be computed as the sum of the Lehmer code (that is similar to the inversion vector) using the factorial number system and can be expressed as \cite{knuth1998volume},
		\begin{align}\label{app_eq1}
		N(\chi_r) = \sum_{j=1}^{M} \!\!\!\mod\bigg(r_j,j\bigg),
		\end{align}
		where $r_j = \left \lfloor r_{j-1}/(j-1)\right \rfloor$ with $r_1=r$ and $\!\!\!\mod(r_j,j)$ denotes the remainder when $r_j$ is divided by $j$. 
		Using the method of induction, it can be shown that $r_{j} = \left \lfloor r/(j-1)!\right \rfloor$. Therefore, the number of inversions given in \eqref{app_eq1} can be written as,
		\begin{align}\label{app_eq2}
		N(\chi_r)  =\!\!\!\mod(r,2)\!+\!2 \left \lfloor r/2\right \rfloor\!+\!\sum_{j=2}^{n} (1\!-\!j) \left \lfloor r/j!\right \rfloor,
		\end{align}
		when $n$ is selected such that $n! \le r < (n+1)!$. Based on this, the sign of the permutation $\chi_r$ can be expressed as \cite{Perm_group},
		\begin{align}\label{app_eq3}
		\textrm{sign}(\chi_r) = (-1)^{\textrm{mod}(r,2)} \times \prod_{j=2}^{n} (-1)^{(1-j) \left \lfloor r/j!\right \rfloor}.
		\end{align}	
		Let us now consider the two adjacent permutations given by the $(2i)$-th permutation and the $(2i+1)$-th permutation. Given that $n!$ is even for $n\ge 2$, we have $n! \le 2i, 2i+1 < (n+1)!, \; \forall i$. Therefore, the signs of these two permutations can be computed using \eqref{app_eq3} as
		\begin{align}\label{app_eq5}
		\textrm{sign}(\chi_{2i}) &= \prod_{j=2}^{n} (-1)^{k_j},\\
		\textrm{sign}(\chi_{2i+1}) &= (-1)\prod_{j=2}^{n} (-1)^{k_j},
		\end{align}
		where $k_j = (1-j)\left \lfloor (2i)/j!\right \rfloor = (1-j)\left \lfloor (2i+1)/j!\right \rfloor, \, \forall j\le n$. Thus, $\textrm{sign}(\chi_{2i+1}) = (-1)\times \textrm{sign}(\chi_{2i})$. This completes the proof of Lemma \ref{Lemma1}.
		
	\end{appendices}
	

\end{document}